\documentclass{article}

\usepackage[preprint]{neurips_2026}

\usepackage[utf8]{inputenc} 
\usepackage[T1]{fontenc}    
\usepackage{hyperref}       
\usepackage{url}            
\usepackage{booktabs}       
\usepackage{amsfonts}       
\usepackage{nicefrac}       
\usepackage{microtype}      
\usepackage{xcolor}         

\usepackage{amsmath}
\usepackage{amsthm}
\usepackage{graphicx}
\usepackage{threeparttable}
\usepackage{multirow}
\usepackage{multicol}
\usepackage{pifont}

\usepackage{wrapfig}
\usepackage{xspace}
\newcommand{\ie}{\textit{i.e.,}\xspace}
\newcommand{\eg}{\textit{e.g.,}\xspace}
\newcommand{\camal}{\textsc{CAMAL}\xspace}
\usepackage{todonotes}
\usepackage{subcaption}

\newtheorem{proposition}{Proposition}

\usepackage{tcolorbox}

\title{CAMAL: Improving Attention Alignment and Faithfulness with Segmentation Masks}

\author{
  Rajdeep Singh Hundal\\
  National University of Singapore\\
  \texttt{rajdeep@u.nus.edu}\\
  \And
  Yan Xiao\\
  Sun Yat-sen University\\
  \texttt{xiaoy367@mail.sysu.edu.cn}\\
  \AND
  Jin Song Dong\\
  National University of Singapore\\
  \texttt{dcsdjs@nus.edu.sg}\\
  \And
  Manuel Rigger\\
  National University of Singapore\\
  \texttt{rigger@nus.edu.sg}\\
}

\begin{document}

\maketitle

\begin{abstract}
Many vision datasets now provide segmentation masks in addition to annotated images to support a wide range of tasks.
In this work, we propose \emph{Class Activation Map Attention Learning} (\camal), an efficient and scalable method that utilizes segmentation masks to improve attention \emph{alignment} and \emph{faithfulness} in vision models.
Specifically, attention alignment refers to the degree to which a model's attention aligns with ground-truth discriminative regions, while attention faithfulness refers to the degree to which a model's attention influences its decision.
Improving both attention alignment and faithfulness is essential for ensuring that model attention is both spatially accurate and causally meaningful.
To improve attention alignment and faithfulness in vision models, \camal first extracts the model's attention for each image during training and then compares the attention to ground-truth discriminative regions obtained from the corresponding segmentation masks.
\camal then acts as an auxiliary regularizer, encouraging attention that aligns with ground-truth discriminative regions, while suppressing attention elsewhere.
We evaluated \camal across two learning paradigms---\emph{Deep Learning} (DL) and \emph{Deep Reinforcement Learning} (DRL)---and observed consistent, significant improvements in both attention alignment and faithfulness.
In particular, CAMAL yields statistically significant gains in attention alignment across all settings, and improves attention faithfulness by over 35\% compared to recent work.
Moreover, we show that improved attention alignment and faithfulness enhance explainability, while yielding improved or comparable generalization performance without increasing inference cost.
These findings demonstrate that the spatial information contained within segmentation masks can be effectively leveraged to guide model attention across learning tasks.
\end{abstract}

\section{Introduction}

As vision models are increasingly deployed in real-world applications, it is imperative to understand their decisions and correct them when necessary.
To this end, \emph{Class Activation Maps} (CAMs)~\citep{zhou2016learning} are widely used to interpret vision models.
They highlight a model's attention, that is, the model's perceived discriminative regions in an image.
This provides an explanation of which regions the model attends to during predictions.
However, as an explainability technique, CAM approaches cannot ensure attention \emph{alignment}---the degree to which a model's attention aligns with ground-truth discriminative regions.
In fact, recent work by~\citet{ye2026the} shows that models are still susceptible to shortcut learning, mistaking spurious regions for discriminative ones.
Furthermore, as demonstrated by~\citet{adebayo2018sanity}, attention alignment alone does not guarantee attention \emph{faithfulness}---the degree to which a model's attention influences its decision.
That is to say, attention might itself be inconsequential to the model's decision-making process.
Thus, CAMs offer interpretation, but lack corrective capability.

To address this gap in corrective capability, recent work have proposed \emph{attention supervision}.
Notably,~\citet{yang2026learning},~\citet{nautiyal2025paric}, and~\citet{petryk2022guiding} proposed auxiliary training objectives that minimize the distance between a model's attention and reference regions that indicate where the attention should be.
A common characteristic of these works is that they all use \emph{CLIP-based priors}---pretrained zero-shot models originally proposed by OpenAI~\citep{radford2021learning}---to derive the reference regions.
We term these reference regions as \emph{pseudo-discriminative regions}, a proxy for ground-truth discriminative regions.
Crucially, because these priors originate from pretrained zero-shot models optimized for broad generalization rather than precise spatial localization, the reference regions they project may only coarsely approximate true discriminative regions~\citep{li2022expansion}, a phenomenon we empirically observed in a pilot study (see~\autoref{app:eed}).
Consequently, this leads to attention being corrected with signals that are themselves \emph{untrustworthy}.

Many vision datasets now provide segmentation masks in addition to annotated images to support a wide range of tasks, including those used in \emph{Deep Learning} (DL)~\citep{illarionova2025packed,OpenImagesV7,CelebAMask-HQ} and \emph{Deep Reinforcement Learning} (DRL)~\citep{li2023behavior,Wydmuch2019ViZdoom,dosovitskiy2017carla}.
Furthermore, modern developments such as SAM 3 by Meta~\citep{carion2025sam3segmentconcepts} significantly reduce the cost of obtaining high-quality segmentation masks by enabling prompt-based automatic segmentation.
Crucially, these segmentation masks are typically human-verified, and thus serve as a \emph{trustworthy} signal that reflects ground-truth discriminative regions.
This increased accessibility to segmentation masks raises a natural question: can the well-structured, complete, and accurate spatial information provided within these masks be leveraged to reliably correct attention in vision models?

In this work, we propose \emph{Class Activation Map Attention Learning} (\camal) as an efficient and scalable approach towards attention supervision that leverages segmentation masks to reliably correct attention in vision models.
Specifically, \camal leverages these masks to identify and encourage attention within ground-truth discriminative regions, while suppressing attention in spurious regions.
At its core, \camal is based on the premise that attention supervision must (1) be grounded in trustworthy reference regions to ensure reliability, and (2) scale efficiently.
We provide our source code at \href{https://doi.org/10.5281/zenodo.20084662}{https://doi.org/10.5281/zenodo.20084662} to support reproducibility.
Our contributions are:
\begin{enumerate}
    \item We introduce \camal, a principled approach to attention supervision that grounds attention learning in trustworthy reference regions obtained from segmentation masks.
    \item We formulate a mask-guided auxiliary regularization objective that jointly encourages attention within ground-truth discriminative regions and suppresses attention in spurious regions.
    In addition, we propose an efficient batch-level attention extraction method that reduces the typical computational overhead from $O(B)$ to $O(1)$, enabling scalable supervision.
    \item We provide a comprehensive and statistically rigorous evaluation of attention supervision across two learning paradigms, DL and DRL.
    In our evaluation, we show that \camal achieves consistent, statistically significant improvements in attention alignment and over 35\% gains in attention faithfulness, while maintaining improved or comparable generalization across 18 DL settings and 4 DRL environments.
\end{enumerate}

\section{Related Work}
Recent work by~\citet{yang2026learning},~\citet{nautiyal2025paric}, and~\citet{petryk2022guiding} improve attention through auxiliary regularization objectives that align attention with reference regions (see~\autoref{app:eed} for a more detailed discussion of these works).
Beyond auxiliary regularization objectives, several other approaches have been explored for improving attention in vision models.
Among these, input augmentation approaches improve attention by modifying training images.
For example,~\citet{Singh_2017_ICCV} randomly masked patches in training images, while~\citet{yun2019cutmix} swapped patches instead.
Architectural approaches in the same vein have also been explored.
For example,~\citet{zhang2018adversarial} augmented the network architecture by introducing an additional classifier head that takes masked feature maps as input, while~\citet{choe2019attention} introduced a self-attention mechanism for each feature map.
Both input and architectural augmentation approaches encourage the model to discover a broader set of ground-truth discriminative regions instead of relying on the most discriminative regions.
\camal shares a similar motivation, but instead improves attention directly through an auxiliary regularization objective grounded in trustworthy reference regions, rather than indirectly through input or network augmentation.

\section{Method}
\label{sec:method}
\begin{figure}[t]
    \centering
    \includegraphics[width=\linewidth,trim=0.9cm 0cm 4.2cm 0cm,clip]{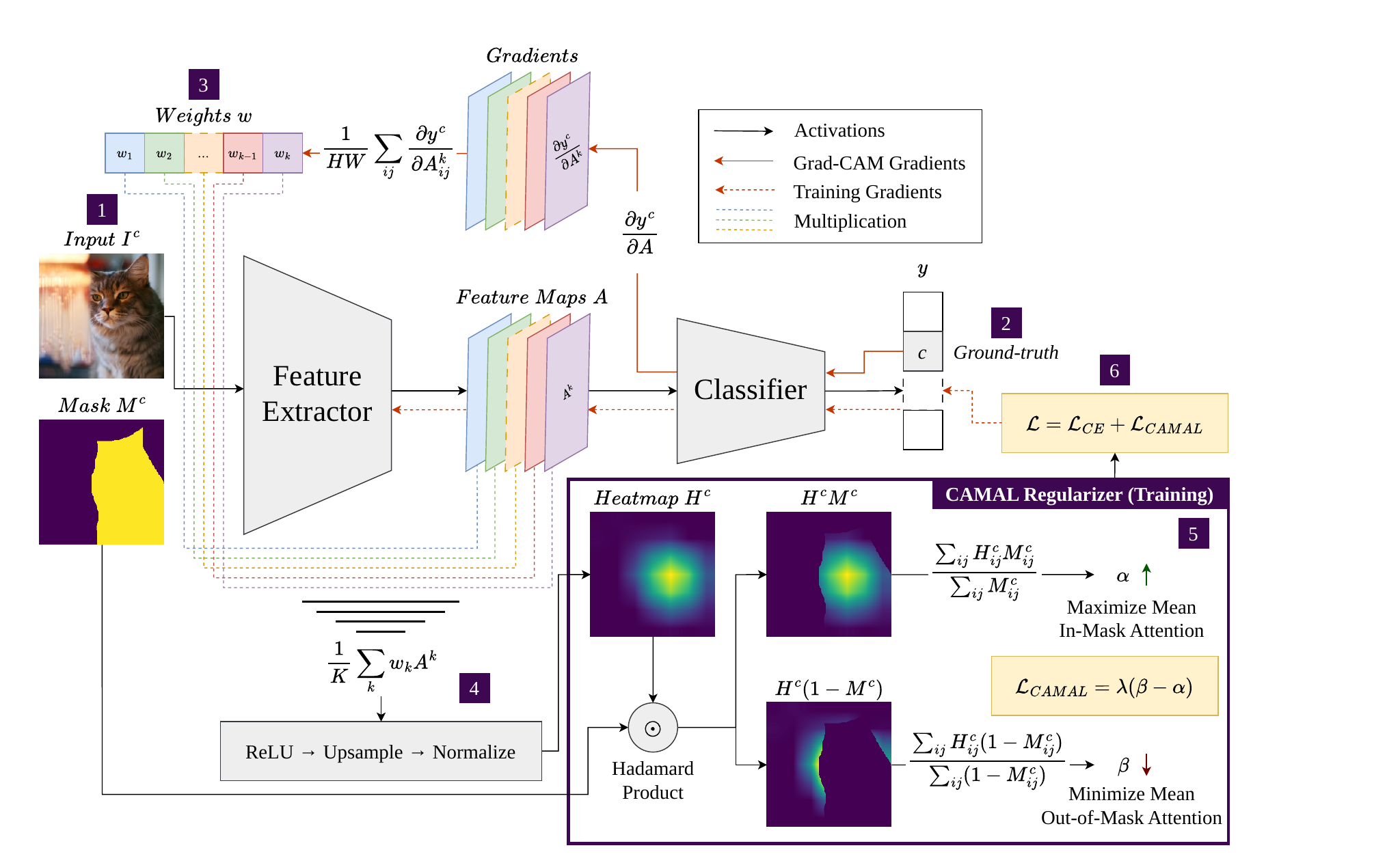}
    \caption{Architecture diagram for \camal. Given an input image $I^c$ passed through a model (steps 1--2), \camal extracts the model's attention $H^c$ (steps 3--4) and compares it with ground-truth discriminative regions obtained from the corresponding segmentation mask $M^c$ (step 5). Lastly, \camal regularizes the network by encouraging attention within ground-truth discriminative regions ($\alpha\uparrow$) while suppressing attention in spurious regions ($\beta\downarrow$) (steps 5--6).
    Importantly, while illustrated with a \emph{Convolutional Neural Network} (CNN) for clarity, \camal can be naturally extended to a wide range of network architectures and tasks.}
    \label{fig:arch}
\end{figure}

We ground attention supervision in trustworthy reference regions obtained directly from segmentation masks, formulating a mask-guided auxiliary regularization objective that explicitly encourages attention within ground-truth discriminative regions while suppressing attention in spurious regions (see~\autoref{fig:arch}).
In this section, we first discuss the proposed auxiliary regularization objective, followed by batch-level optimizations to make attention supervision feasible for realistic training regimes, and finally describe how \camal can be naturally extended to a wide range of network architectures and tasks.

\subsection{\camal Regularization Objective}
At a high level, the \camal regularization objective quantifies and promotes agreement with ground-truth discriminative regions while penalizing agreement with spurious regions.
To formalize this regularization objective, let $I^c \in \mathbb{R}^{3 \times H \times W}$ denote an input image with ground-truth class $c$.
We define the model's attention for $I^c$ as a spatial heatmap $H^c \in [0,1]^{H \times W}$ that captures the model's perceived discriminative regions for class $c$ (steps 1--4 in~\autoref{fig:arch}).
In this work, we instantiate $H^c$ with Grad-CAM~\citep{selvaraju2017grad}, a gradient-based technique that extracts attention.
However, \camal is not tied to Grad-CAM and can operate with any attention extraction technique.
Correspondingly, let $M^c \in \{0,1\}^{H \times W}$ denote the segmentation mask for $I^c$, where $M_{ij} = 1$ indicates the ground-truth discriminative regions for class $c$.
Our objective is to encourage the model's attention $H^c$ to align with $M^c$ (steps 5--6 in~\autoref{fig:arch}).
Intuitively, this biases the model to rely on ground-truth discriminative regions enclosed by the segmentation masks when making decisions, rather than spurious regions.

To realize this objective, we decompose a model's attention $H^c$ into two components: attention within ground-truth discriminative regions and attention outside them.
Specifically, we define
\begin{equation}
    \alpha = \frac{\sum\limits_{ij}H^c_{ij}M^c_{ij}}{\sum\limits_{ij}M^c_{ij}}
    \quad\text{and}\quad
    \beta = \frac{\sum\limits_{ij}H^c_{ij}(1-M^c_{ij})}{\sum\limits_{ij}(1-M^c_{ij})}
\end{equation}
where $\alpha$ measures the mean attention within ground-truth discriminative regions, while $\beta$ measures the mean attention outside them (step 5 in~\autoref{fig:arch}).
Recent work on attention supervision typically employ regularizers that focus solely on suppressing attention in spurious regions, that is, by minimizing terms analogous to $\beta$~\citep{nautiyal2025paric,petryk2022guiding}.
However, such formulations do not explicitly account for attention within ground-truth discriminative regions and thus, fail to capture the relative discrepancy between model attention and the ground-truth.
We empirically demonstrate this limitation in~\autoref{app:eed}.
In contrast, our formulation jointly considers both $\alpha$ and $\beta$, enabling a direct measure of how far attention deviates from the ground-truth.

A well-behaved model should attend more towards ground-truth discriminative regions ($\alpha\uparrow$) and less towards spurious regions ($\beta\downarrow$).
Accordingly, we define an auxiliary regularization objective that encourages $\alpha$ to be large and $\beta$ to be small by minimizing $\beta - \alpha$: 
\begin{equation}
    \label{eq:camalloss}
    \mathcal{L} = \overbrace{-\frac{1}{B} \sum_{b} \log(p_{b,y_b})}^{\mathcal{L}_{CE}} + \overbrace{\frac{\lambda}{B}\sum_{b} (\beta_b - \alpha_b)}^{\mathcal{L}_{CAMAL}}
\end{equation}
Here, $\mathcal{L}$ is minimized over a mini-batch of size $B$ where $b \in \{1,...,B\}$ and $\mathcal{L}_{CAMAL}$ acts as an auxiliary regularizer for cross-entropy loss $\mathcal{L}_{CE}$ (step 6 in~\autoref{fig:arch}).
$\lambda$ controls \camal's regularizing impact and can be tuned empirically.
Overall, \camal introduces an additional mask-guided auxiliary regularizer to the de facto objective, which emphasizes in-mask attention via $\alpha$ and suppresses out-of-mask attention via $\beta$.

\subsection{Efficient Batch-Level Attention Supervision}
\label{sec:scale}
\begin{wrapfigure}[17]{r}{0.40\textwidth}
    \centering
    \includegraphics[width=\linewidth,trim=0cm 0cm 0cm 0cm,clip]{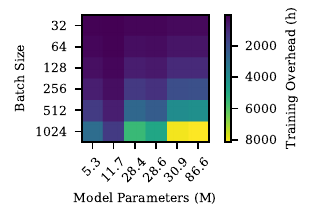}
    \caption{Training overhead incurred by per-sample attention extraction relative to batch-level attention extraction. Values are based on empirical measurements from our experiments where available, with remaining points estimated from known scaling trends.}
    \label{fig:cpe}
\end{wrapfigure}
Existing attention supervision approaches do not scale well to large datasets, as the gradient-based attention extraction techniques they rely on are typically formulated on a per-sample basis.
As shown in~\autoref{fig:cpe}, the training overhead of per-sample attention extraction grows linearly with batch size, quickly becoming prohibitive for larger models and realistic training regimes.
Concretely, we consider Grad-CAM, a representative gradient-based attention extraction technique:
\begin{equation}
    \label{eq:gc}
    G^{c} = ReLU\underbrace{\Bigg(\frac{1}{K}\sum_k \overbrace{\Bigg(\frac{1}{HW} \sum_{ij} \frac{\partial y^c}{\partial A^k_{ij}}\Bigg)}^{w_k} A^k\Bigg)}_{\text{linear combination}}
\end{equation}
Here, the Grad-CAM for an image $I^c$ classified by a CNN is defined by $G^c$ where $c$, $w_k$, and $A^k$ denote the ground-truth class, weight for the $k$-th feature map, and $k$-th feature map respectively (steps 1--4 in~\autoref{fig:arch}).
Crucially, to extract $G^c$ for $I^c$ in a mini-batch of size $B$, a single isolated backward pass from the corresponding scalar logit $y^c$ is required.
This translates to an additional $B$ backward passes per mini-batch.
To the best of our knowledge, recent work in attention supervision leveraging gradient-based attention extraction techniques do not explicitly formalize how to address this bottleneck and implement their approach efficiently at scale~\citep{yang2026learning,nautiyal2025paric,petryk2022guiding}.
However, this is crucial as attention supervision must be performed at every training step, rendering per-sample attention extraction computationally inefficient, as shown in~\autoref{fig:cpe}.

To address this, we introduce and formalize an efficient, scalable approach towards extracting attention for supervision with gradient-based techniques.
Specifically, we sum the corresponding $y^c$ logits across the mini-batch into a single scalar and perform a backward pass with respect to this scalar to extract attention for all samples in the mini-batch.
Although counterintuitive, as this might seem like mixing attention signals across the mini-batch, this is equivalent to computing attention for each sample individually, while reducing the computational complexity from $O(B)$ to $O(1)$ (see~\autoref{prop1} in~\autoref{app:eed} for proof).
Using Grad-CAM in~\autoref{eq:gc} to illustrate this operation, the weight $w_k$ for the $k$-th feature map $A^k$ is reformulated as
\begin{equation}
\underbrace{\frac{1}{HW} \sum_{ij} \frac{\partial y^c}{\partial A^k_{ij}}}_{w_k} \rightarrow
\underbrace{\frac{1}{HW}\sum_{ij} \Bigg(\frac{\partial}{\partial A^k_{ij}}\overbrace{\sum_{b} y_b^{c_b}}^{\text{summed scalar}}\Bigg)}_{\text{reformulated }w_k}
\end{equation}
where $b \in \{1,...,B\}$ denotes samples in a mini-batch and the gradients for a sample are computed with respect to the summed scalar.
Furthermore, this formulation naturally extends to other gradient-based attention extraction techniques, such as DiffGradCAM~\citep{piland2025diffgradcam}, LayerCAM~\citep{jiang2021layercam}, Integrated Gradients~\citep{sundararajan2017axiomatic}, and SmoothGrad~\citep{smilkov2017smoothgrad}.

\subsection{Architecture and Task Agnosticism}
While presented in~\autoref{fig:arch} with respect to a CNN-based classifier for clarity, \camal can be extended to any network architecture as long as spatial feature maps can be extracted from intermediate representations.
For example, transformers can be naturally adapted to work with \camal.
Let $X \in \mathbb{R}^{B \times N \times C}$ denote the transformer output where $B$, $N$, and $C$ correspond to the batch, sequence, and channel dimensions respectively.
Assuming patch-based tokenization, $X$ can be reshaped into $X \in \mathbb{R}^{B \times H \times W \times C}$, where $H$ and $W$ correspond to the height and width, yielding spatial feature maps analogous to those produced by CNNs.
However, the exact semantics of this transformation might differ depending on the network architecture.
For example, $N = H \times W + 1$ for ViT by~\citet{dosovitskiy2021an}, where the additional token corresponds to the class token and is omitted from the transformation.
In contrast, for Swin-V2 by~\citet{liu2022swin}, $N$ is already in the form of $H \times W$ and thus, no transformation is required.
This transformation can also be ignored for hybrid architectures like MaxViT by~\citet{tu2022maxvit}, where convolutional stages already produce explicit spatial feature maps.

Furthermore, \camal extends to any task with a well-defined scalar target from which attention can be computed.
For example, in our experiments, we evaluate \camal in a weakly-supervised setting with DRL.
Specifically, we augment the \emph{Proximal Policy Optimization} (PPO) objective~\citep{schulman2017proximal} as $\mathcal{L} = \mathcal{L}_{PPO} + \mathcal{L}_{CAMAL}$.
Unlike supervised learning, DRL does not provide ground-truth labels---there is no known correct action for a given state $s$.
Instead, the actor network outputs a probability distribution over actions and the policy learns to select actions by maximizing the expected return.
As this does not yield a scalar supervision signal analogous to a ground-truth logit $y^c$, we instead use the critic's value estimate $V(s)$---the expected return from state $s$---as the target for attention computation.

\section{Experiments}
\label{sec:exp}
We compare the effectiveness of \camal in attention alignment, faithfulness, and generalization against two baselines, \emph{Prior} and \emph{Vanilla}.
The former uses pseudo-discriminative regions derived from CLIP as reference regions (\ie similar to recent work), while the latter does not incorporate any form of attention supervision, and is trained solely on the task objective.
Specifically, we trained six DL models of varying sizes and layer compositions on three challenging fine-grained classification datasets, where classes are highly similar and require learning subtle visual distinctions.
The evaluated models comprise ResNet-18~\citep{he2016deep}, ConvNeXt-T~\citep{liu2022convnet}, EfficientNet-B0~\citep{tan2019efficientnet}, ViT-B/16~\citep{dosovitskiy2021an}, SwinV2-T~\citep{liu2022swin}, and MaxViT-T~\citep{tu2022maxvit}.
The evaluated datasets comprise BUSI~\citep{ALDHABYANI2020104863}, Oxford-IIIT-Pet~\citep{parkhi12a}, and Oxford-Flower~\citep{Nilsback08}.

In addition, we explore the effectiveness of attention supervision in a weakly-supervised setting by extending the comparison of \camal and Vanilla to DRL across four ViZDoom~\citep{Kempka2016ViZDoom} environments under the PPO objective.
In particular, ViZDoom is a visual DRL benchmark based on the Doom game engine, featuring first-person shooter environments.
We select the Basic, DefendLine, DefendCenter, and TakeCover environments to cover diverse task demands, ranging from single-target combat to multi-threat engagement and pure evasive control.

For the DL experiments, we use 10-fold cross-validation for each model-dataset-method combination, compute 95\% \emph{confidence intervals} (CI) with mean $\pm$ 1.96 $\times$ standard error, and assess statistical significance with the \emph{Wilcoxon Signed-rank Test} (WSRT)~\citep{rainio2024evaluation,demvsar2006statistical}.
For the DRL experiments, we conduct 5 trials for each environment-method combination, analyze the results with \emph{Stratified Bootstrap Confidence Intervals} (SBCI), and assess statistical significance with the \emph{Probability of Improvement} (POI)~\citep{agarwal2021deep}.
Both SBCI and POI are computed using the open-source library by~\citet{agarwal2021deep}.
More details on the statistical methods and experimental settings can be found in~\autoref{app:sem} and~\autoref{app:esd} respectively.

\subsection{Attention Alignment}
We quantitatively assess attention alignment by measuring the extent to which model attention overlaps with ground-truth discriminative regions by computing their \emph{Intersection-over-Union} (IoU), a well-established metric widely used in computer vision tasks.
This addresses two key limitations in recent work.
The first is the lack of quantitative evaluation for attention alignment~\citep{yang2026learning,nautiyal2025paric}, potentially favoring methods that produce visually compelling explanations~\citet{adebayo2018sanity}.
The second is the reliance on the pointing game experiment~\citep[Figure 5]{petryk2022guiding}, which considers attention correct as long as the pixel with the highest attention lies within ground-truth discriminative regions, thereby overlooking attention in spurious regions.

\begin{table}[t]
    \centering
    \begin{threeparttable}
    \caption{Comparison of attention alignment between \camal and Prior in DL. The evaluation metric is test IoU (mean $\pm$ standard deviation). Results in bold denote cases where \camal achieves a higher mean than Prior, with the difference being statistically significant under WSRT ($p < 0.05$).}
        \begin{tabular}{@{}lrrrrrr@{}}
            \toprule
            \multirow{2.6}{*}{Model} & \multicolumn{2}{c}{BUSI} & \multicolumn{2}{c}{Oxford-IIIT-Pet} & \multicolumn{2}{c}{Oxford-Flower} \\
            \cmidrule(lr){2-3} \cmidrule(lr){4-5} \cmidrule(lr){6-7}
            & \multicolumn{1}{c}{CAMAL} & \multicolumn{1}{c}{Prior} & \multicolumn{1}{c}{CAMAL} & \multicolumn{1}{c}{Prior} & \multicolumn{1}{c}{CAMAL} & \multicolumn{1}{c}{Prior}\\
            \midrule
            ResNet & \textbf{33.8 $\pm$ 1.9} & 2.2 $\pm$ 0.6 & \textbf{67.1 $\pm$ 1.5} & 2.5 $\pm$ 0.6 & \textbf{69.4 $\pm$ 1.1} & 5.3 $\pm$ 0.6 \\

            ConvNeXt & \textbf{35.6 $\pm$ 3.8} & 2.7 $\pm$ 0.5 & \textbf{71.3 $\pm$ 0.6} & 3.8 $\pm$ 0.5 & \textbf{71.0 $\pm$ 1.1} & 9.9 $\pm$ 2.0 \\
            
            EfficientNet & \textbf{35.9 $\pm$ 2.9} & 2.6 $\pm$ 0.8 & \textbf{71.3 $\pm$ 0.7} & 4.4 $\pm$ 0.8 & \textbf{71.7 $\pm$ 0.8} & 9.6 $\pm$ 0.7 \\
            
            ViT & \textbf{42.8 $\pm$ 5.5} & 1.8 $\pm$ 1.1 & \textbf{80.2 $\pm$ 0.4} & 3.2 $\pm$ 0.4 & \textbf{80.7 $\pm$ 0.9} & 10.1 $\pm$ 1.0 \\
            
            Swin & \textbf{38.0 $\pm$ 3.0} & 2.8 $\pm$ 0.6 & \textbf{73.2 $\pm$ 0.4} & 5.1 $\pm$ 0.4 & \textbf{72.3 $\pm$ 0.6} & 12.5 $\pm$ 1.0 \\
            
            MaxViT & \textbf{38.7 $\pm$ 1.5} & 2.6 $\pm$ 0.4 & \textbf{74.0 $\pm$ 0.4} & 4.8 $\pm$ 1.1 & \textbf{72.3 $\pm$ 0.7} & 11.5 $\pm$ 0.7 \\
            \bottomrule
        \end{tabular}
    \label{tab:cpa}
    \end{threeparttable}
\end{table}
\begin{figure}[t]
    \centering
    \includegraphics[width=0.9\linewidth,trim=0cm 0cm 0cm 0cm,clip]{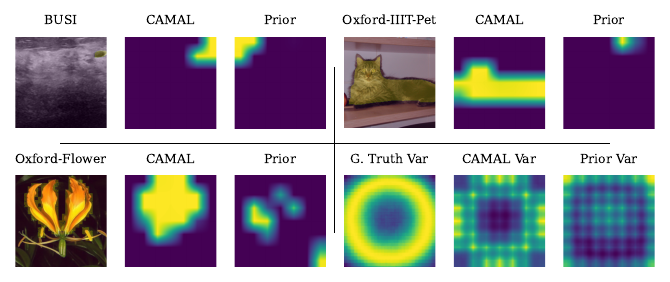}
    \caption{Comparison of attention alignment between \camal and Prior in DL. For each dataset, we present an image overlaid with its ground-truth discriminative regions, alongside the corresponding attention from \camal and Prior. A single representative example per dataset is shown for brevity; however, similar patterns are consistently observed across datasets.}
    \label{fig:cpa}
\end{figure}
Quantitative and qualitative results for the DL experiments are shown in~\autoref{tab:cpa} and~\autoref{fig:cpa} respectively.
\camal consistently outperforms Prior across all models and datasets, achieving substantially higher IoU scores, with all improvements being statistically significant under WSRT.
As shown in~\autoref{fig:cpa}, \camal focuses more on discriminative regions, while Prior often attends to spurious regions.
The bottom-right panel further illustrates this difference by visualizing the spatial variance of attention across datasets.
\camal exhibits structured variance concentrated around object boundaries, aligning with the ground truth. This produces a roughly circular pattern, as objects are typically centered and variations occur primarily along their edges.
In contrast, Prior concentrates variance in peripheral regions, suggesting that attention is often allocated to spurious areas.
We omit the DL comparison with Vanilla for brevity, as it exhibits quantitatively and qualitatively similar trends, with \camal consistently outperforming Vanilla (see~\autoref{app:eed} for full results).

\begin{figure*}[t]
    \centering
    \hspace{0.08\linewidth}
    \begin{subfigure}{0.40\linewidth}
    \includegraphics[width=0.8\linewidth,trim=0cm 0cm 0cm 0cm,clip]{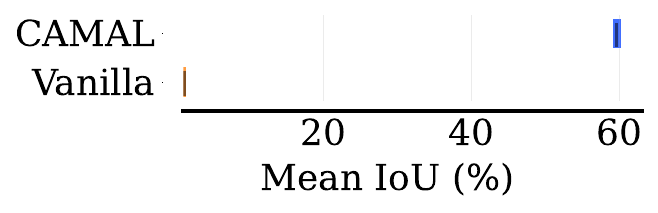}
    \end{subfigure}
    \begin{subfigure}{0.5\linewidth}
    \includegraphics[width=0.8\linewidth,trim=0cm 0cm 0cm 0cm,clip]{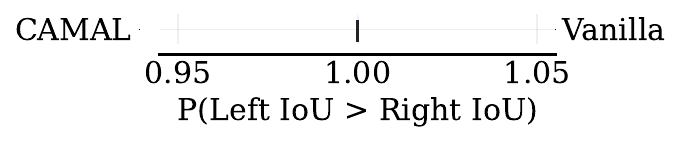}
    \end{subfigure}
    \caption{Comparison of attention alignment between \camal and Vanilla in DRL. The left figure uses SBCI to statistically analyze the results across environments and trials while the right figure uses POI to determine if the observed differences are statistically significant ($P(\text{Left IoU}>\text{Right IoU}) > 0.5 \; \land \; 0.5 \notin CI$). The shaded regions denote 95\% confidence bands.}
    \label{fig:cvda}
\end{figure*}
Quantitative results for the DRL experiments are shown in~\autoref{fig:cvda} (see~\autoref{app:eed} for qualitative results).
Under SBCI analysis, where results across multiple environments and trials are repeatedly sampled and aggregated to obtain point estimates with corresponding 95\% CIs, \camal achieves a substantially higher mean IoU than Vanilla, indicating improved attention alignment under the PPO objective.
Furthermore, the POI analysis confirms that this difference in IoU is statistically significant, with both the point estimate and the lower bound of the CI exceeding 0.5.
We omit including a comparison with Prior for DRL as its reliance on pseudo-discriminative regions makes it impractical in an online learning setting, where such regions will need to be recomputed for newly generated observations, significantly increasing inference overhead.

\begin{tcolorbox}[sharp corners]
    \textbf{Key takeaway:} These findings show that \camal reliably yields better-aligned attention in both DL and DRL, consistently outperforming Prior and Vanilla, underscoring that aligned attention is best achieved when attention supervision is grounded in trustworthy regions.
\end{tcolorbox}

\subsection{Attention Faithfulness}
We quantitatively assess attention faithfulness to determine whether the model's attention reflects features that truly influence its predictions.
Specifically, we assess the causal impact of attention by progressively \emph{removing} and \emph{inserting} pixels in order of their attention-attributed importance and measuring the resulting change in model confidence (\ie softmax output).
Here, removal starts from the original image and progressively masks pixels to zero in order of importance, while insertion starts from an all-zero image and progressively restores pixels in the same order.
If the attention is faithful, removing highly attended regions should significantly decrease model confidence, while inserting them should significantly increase it.
This evaluation protocol follows well-established practices in \emph{Explainable Artificial Intelligence} (XAI)~\citep{fong2019understanding,Petsiuk2018rise,zeiler2014visualizing}; however, similar to attention alignment, recent work on attention supervision has yet to incorporate such quantitative evaluations of attention faithfulness.

\begin{figure}[t]
    \centering
    \includegraphics[width=0.9\linewidth,trim=0cm 0cm 0cm 0cm,clip]{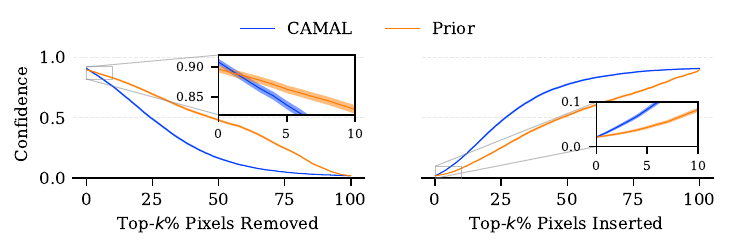}
    \caption{Comparison of attention faithfulness between \camal and Prior in DL. Model confidence is evaluated under removal and insertion of the top-$k$\% most important pixels determined by the model's attention. Curves are aggregated over all model-dataset combinations, with shaded regions denoting 95\% confidence bands.}
    \label{fig:cpf}
\end{figure}
The results for the DL experiments are shown in~\autoref{fig:cpf}.
\camal consistently exhibits more faithful attention than Prior under both removal and insertion.
Under removal, model confidence drops substantially faster for \camal, indicating that its attention more accurately identifies features critical to the prediction.
In contrast, Prior shows a more gradual decline, suggesting that its attention is less concentrated on truly influential regions.
This is further reflected in the \emph{area under the curve} (AUC), where \camal achieves 39\% lower AUC than Prior.
Under insertion, \camal achieves a faster recovery in model confidence as pixels are restored, further demonstrating that its attention captures features that are causally relevant to the model's output.
Consistently, \camal attains 35.3\% higher AUC than Prior under insertion.
These trends are most pronounced where $k < 50$, which typically corresponds to the object region; this expected behavior is clearly followed by \camal but not by Prior.
Similar to attention alignment, we omit the DL comparison with Vanilla for brevity, as it exhibits similar trends, with \camal consistently outperforming Vanilla (see~\autoref{app:eed} for full results).
Furthermore, we do not include a comparison for DRL here as the model's output may depend on future states, making per-frame removal and insertion less meaningful.

\begin{tcolorbox}[sharp corners]
    \textbf{Key takeaway:} These findings show that \camal reliably yields more faithful attention, consistently outperforming Prior and Vanilla, further demonstrating that faithful attention is best achieved when attention supervision is grounded in trustworthy regions.
\end{tcolorbox}

\subsection{Generalization Performance}
We quantitatively assess whether attention supervision improves generalization reliably by evaluating classification accuracy for DL and mean return for DRL.
\begin{table}[t]
    \centering
    \begin{threeparttable}
    \caption{Comparison of generalization between \camal and Vanilla in DL. The evaluation metric is test accuracy (mean $\pm$ standard deviation). Results in bold denote cases where \camal achieves a higher mean than Vanilla, with the difference being statistically significant under WSRT ($p < 0.05$). Results that are underlined denote cases where \camal achieves a higher mean without statistical significance.}
        \begin{tabular}{@{}lrrrrrr@{}}
            \toprule
            \multirow{2.6}{*}{Model} & \multicolumn{2}{c}{BUSI} & \multicolumn{2}{c}{Oxford-IIIT-Pet} & \multicolumn{2}{c}{Oxford-Flower} \\
            \cmidrule(lr){2-3} \cmidrule(lr){4-5} \cmidrule(lr){6-7}
            & \multicolumn{1}{c}{CAMAL} & \multicolumn{1}{c}{Vanilla} & \multicolumn{1}{c}{CAMAL} & \multicolumn{1}{c}{Vanilla} & \multicolumn{1}{c}{CAMAL} & \multicolumn{1}{c}{Vanilla}\\
            \midrule
            ResNet & \textbf{84.1 $\pm$ 3.0} & 80.1 $\pm$ 3.7 & \textbf{78.2 $\pm$ 1.5} & 73.9 $\pm$ 4.0 & \textbf{92.4 $\pm$ 1.0} & 90.1 $\pm$ 2.5 \\

            ConvNeXt & \textbf{83.8 $\pm$ 4.4} & 71.9 $\pm$ 2.5 & \textbf{81.5 $\pm$ 1.8} & 72.8 $\pm$ 3.8 & \textbf{93.5 $\pm$ 1.7} & 92.6 $\pm$ 1.1 \\
            
            EfficientNet & \underline{85.6 $\pm$ 4.0} & 83.3 $\pm$ 6.3 & \textbf{85.6 $\pm$ 1.2} & 83.4 $\pm$ 2.2 & \underline{95.6 $\pm$ 1.0} & 94.9 $\pm$ 1.5 \\
            
            ViT & \textbf{82.3 $\pm$ 2.8} & 70.3 $\pm$ 10.7 & 84.0 $\pm$ 3.9 & 86.6 $\pm$ 3.1 & \textbf{93.8 $\pm$ 1.0} & 88.6 $\pm$ 9.9 \\
            
            Swin & \textbf{87.8 $\pm$ 3.2} & 85.3 $\pm$ 5.5 & \textbf{90.8 $\pm$ 1.3} & 89.4 $\pm$ 1.6 & \underline{97.5 $\pm$ 1.1} & 97.0 $\pm$ 1.2 \\
            
            MaxViT & \underline{89.6 $\pm$ 3.0} & 89.0 $\pm$ 2.3 & \underline{93.9 $\pm$ 1.2} & 93.8 $\pm$ 0.9 & 98.2 $\pm$ 0.7 & 98.6 $\pm$ 0.5 \\
            \bottomrule
        \end{tabular}
    \label{tab:cvg}
    \end{threeparttable}
\end{table}
The results for the DL experiments against Vanilla are shown in~\autoref{tab:cvg}.
\camal consistently outperforms Vanilla, achieving statistically significant improvements in classification accuracy across the majority of the model-dataset combinations.
Even in cases where improvements are not statistically significant, \camal remains competitive, never exhibiting substantial degradation relative to Vanilla.
This shows that attention supervision is indeed effective in improving generalization, but its impact may vary across architectures and datasets, potentially due to how features and spatial dependencies are represented.
In such cases, the task objective alone may already be sufficient to learn reasonably well-aligned and faithful representations, thereby reducing the benefit of attention supervision.

\begin{figure*}[t]
    \centering
    \hspace{0.08\linewidth}
    \begin{subfigure}{0.40\linewidth}
    \includegraphics[width=0.8\linewidth,trim=0cm 0cm 0cm 0cm,clip]{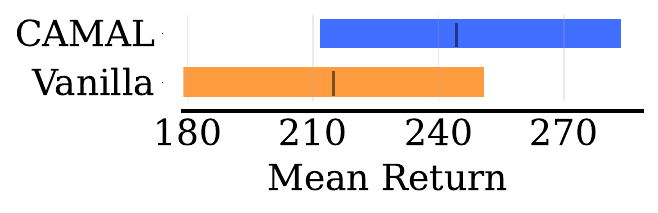}
    \end{subfigure}
    \begin{subfigure}{0.5\linewidth}
    \includegraphics[width=0.8\linewidth,trim=0cm 0cm 0cm 0cm,clip]{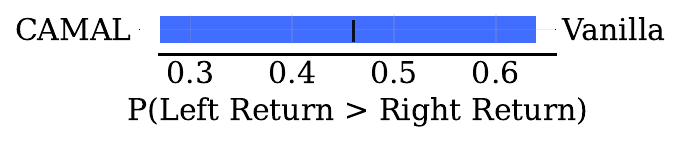}
    \end{subfigure}
    \caption{Comparison of generalization between \camal and Vanilla in DRL. The left figure uses SBCI to statistically analyze the results across environments and trials while the right figure uses POI to determine if the observed differences are statistically significant ($P(\text{Left Return}>\text{Right Return}) > 0.5 \; \land \; 0.5 \notin CI$). The shaded regions denote 95\% confidence bands.}
    \label{fig:cvdg}
\end{figure*}
The results for the DRL experiments against Vanilla are shown in~\autoref{fig:cvdg}.
Using SBCI analysis, \camal attains a higher mean return than Vanilla across environments and trials.
This is, however, not statistically significant under POI, indicating that while \camal shows a trend towards improved performance, the gains are not consistently reliable across environments.
This suggests that attention supervision requires further refinement in more complex tasks, such as accounting for the underlying dynamics and reward structure in DRL.
\begin{table}[t]
    \centering
    \begin{threeparttable}
    \caption{Comparison of generalization between \camal and Prior in DL. The evaluation metric is test accuracy (mean $\pm$ standard deviation). Results in bold denote cases where \camal achieves a higher mean than Prior, with the difference being statistically significant under WSRT ($p < 0.05$). Results that are underlined denote cases where \camal achieves a higher mean without statistical significance.}
        \begin{tabular}{@{}lrrrrrr@{}}
            \toprule
            \multirow{2.6}{*}{Model} & \multicolumn{2}{c}{BUSI} & \multicolumn{2}{c}{Oxford-IIIT-Pet} & \multicolumn{2}{c}{Oxford-Flower} \\
            \cmidrule(lr){2-3} \cmidrule(lr){4-5} \cmidrule(lr){6-7}
            & \multicolumn{1}{c}{CAMAL} & \multicolumn{1}{c}{Prior} & \multicolumn{1}{c}{CAMAL} & \multicolumn{1}{c}{Prior} & \multicolumn{1}{c}{CAMAL} & \multicolumn{1}{c}{Prior}\\
            \midrule
            ResNet & \underline{83.8 $\pm$ 3.2} & 82.1 $\pm$ 4.4 & \textbf{77.8 $\pm$ 2.6} & 73.1 $\pm$ 1.5 & \textbf{92.3 $\pm$ 1.0} & 89.5 $\pm$ 1.8 \\

            ConvNeXt & \textbf{82.3 $\pm$ 4.3} & 76.0 $\pm$ 3.5 & \textbf{80.6 $\pm$ 2.1} & 75.5 $\pm$ 2.4 & \underline{92.9 $\pm$ 2.4} & 91.7 $\pm$ 2.5 \\
            
            EfficientNet & \underline{87.2 $\pm$ 5.1} & 85.0 $\pm$ 4.2 & 85.4 $\pm$ 1.4 & 85.7 $\pm$ 1.3 & 95.6 $\pm$ 0.8 & 96.1 $\pm$ 1.1 \\
            
            ViT & \underline{83.6 $\pm$ 5.4} & 82.6 $\pm$ 5.2 & 85.5 $\pm$ 1.7 & 85.9 $\pm$ 1.5 & 92.5 $\pm$ 2.9 & 94.2 $\pm$ 1.1 \\
            
            Swin & 85.5 $\pm$ 4.5 & 85.5 $\pm$ 4.4 & 90.6 $\pm$ 1.4 & 91.5 $\pm$ 1.1 & 97.3 $\pm$ 0.9 & 98.1 $\pm$ 0.4 \\
            
            MaxViT & 89.9 $\pm$ 3.9 & 90.3 $\pm$ 3.4 & 94.1 $\pm$ 1.3 & 94.3 $\pm$ 1.1 & \underline{98.8 $\pm$ 0.3} & 98.6 $\pm$ 0.4 \\
            \bottomrule
        \end{tabular}
    \label{tab:cpg}
    \end{threeparttable}
\end{table}
Finally, the results for the DL experiments against Prior are shown in~\autoref{tab:cpg}.
\camal performs on par with Prior across models and datasets, achieving higher classification accuracies in half of the model-dataset combinations.

\begin{tcolorbox}[sharp corners]
    \textbf{Key takeaway:} These findings show that \camal is comparable to Prior and often outperforms Vanilla in generalization, suggesting that attention supervision is beneficial for generalization regardless of how reference regions are obtained.
    Furthermore, the comparable performance between \camal and Prior in generalization suggests that, among methods incorporating attention supervision, the primary gains may lie more in improving attention quality than in substantially increasing task performance.
\end{tcolorbox}

\section{Conclusion}
We have proposed \camal, a principled approach to attention supervision that leverages segmentation masks as trustworthy reference regions to guide model attention.
In addition, we introduced an efficient batch-level attention extraction method that enables scalable supervision.
\camal consistently improves attention alignment and faithfulness, while yielding improved or comparable generalization performance.
Importantly, our results suggest that the primary benefits of attention supervision lie in improving the quality of learned attention, with gains in task performance being more modest and context-dependent.
Taken together, these findings highlight that realizing the full potential of attention supervision requires grounding it in trustworthy reference regions, establishing a reliable baseline for future attention supervision methods.

\bibliographystyle{plainnat}
\bibliography{references}


\appendix
\section{Statistical Evaluation Methods}
\label{app:sem}
\paragraph{Wilcoxon Signed-rank Test.}
\citet{rainio2024evaluation} and~\citet{demvsar2006statistical} both proposed using the \emph{Wilcoxon Signed-rank Test} (WSRT)---a non-parametric test---to statistically analyze multiple paired trials between two DL algorithms for more reliable claims.
In particular, given $k$-fold cross-validation where two algorithms use the same dataset split each fold, WSRT tests if the median of the differences in the $k$ paired trials $d_i$, where $i \in \{1,...,k\}$, is zero with the following test statistic:
\begin{equation}
   T = min(R^+, R^-), \;\text{where}\; R^+ = \sum_{d_i > 0} rank(|d_i|) \;\; \text{and} \;\; R^- = \sum_{d_i < 0} rank(|d_i|)
\end{equation}
Here, the absolute differences $|d_i|$ are first ranked collectively.
Subsequently, the ranks are summed for $d_i > 0$ and $d_i < 0$ individually.
The minimum is then taken as the test statistic to be compared with the respective critical values for WSRT.
In this work, we use $10$-fold cross-validation for all DL experiments and one-tailed WSRT to evaluate whether \camal outperforms each DL baseline. 
\begin{figure}[h]
    \centering
    \includegraphics[width=\linewidth,trim=1cm 0.2cm 0.7cm 0.4cm,clip]{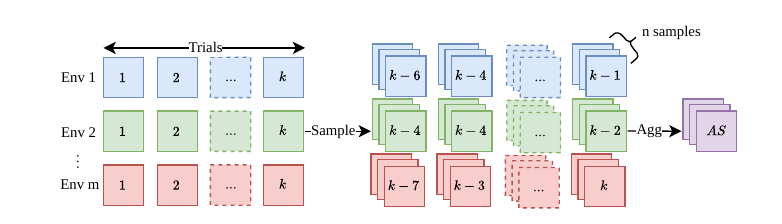}
    \caption{The SBCI evaluation method involves sampling a DRL algorithm's result set repeatedly and aggregating (Agg) the samples to form a distribution of aggregated samples (AS) that estimates the algorithm's true performance.}
    \label{fig:sbci}
\end{figure}
\paragraph{Stratified Bootstrap Confidence Intervals.}
DRL is inherently stochastic and more robust statistical methods are required to reliably compare two DRL algorithms.
\citet{agarwal2021deep} proposed \emph{Stratified Bootstrap Confidence Intervals} (SBCI) to compare two DRL algorithms across multiple environments and trials, illustrated in~\autoref{fig:sbci}.
Given a $m \times k$ result set for a DRL algorithm trained on $m$ environments $k$ times, SBCI samples $n$ sets of shape $m \times k$ from the result set with replacement, with equal proportions per environment.
The $n$ sets are then individually aggregated according to an aggregation metric to form a distribution of aggregated samples with a 95\% CI that reliably estimates the true performance of the DRL algorithm.
In this work, we use SBCI with the mean aggregation metric, where $m=4$ and $k=5$, for all DRL experiments to evaluate whether \camal outperforms each DRL baseline, unless otherwise specified.

\paragraph{Probability of Improvement.}
To determine whether the observed differences under SBCI with mean aggregation are statistically significant, we further conduct SBCI with the \emph{Probability of Improvement} (POI) aggregation metric~\citep{agarwal2021deep}.
POI performs pairwise comparisons using the Mann-Whitney U-statistic~\citep{mann1947test}:
\begin{equation}
    \label{eq:mwu}
   P(X_m > Y_m) = \frac{1}{k^2} \sum_{i=1}^{k}\sum_{j=1}^{k} \mathbb{I}(x_{m,i},y_{m,j}),\;\;
   \text{where} \;\; \mathbb{I}(x,y) = 
   \begin{cases}
       1, & \text{if}\;\; x>y, \\
       \frac{1}{2}, & \text{if}\;\; x=y, \\
       0, & \text{if}\;\; x<y.
   \end{cases}
\end{equation}
Here, $P(X_m>Y_m)$ denotes the POI of algorithm X over Y on environment $m$, with $k$ denoting the number of trials.
The intuition behind the pairwise comparisons being that for X to be considered better than Y, it must outperform Y with sufficient frequency.
$P(X > Y)$ is then obtained by aggregating the environment specific POI's $\frac{1}{M}\sum_{m=1}^{M}P(X_m>Y_m)$.
Finally, the resulting point estimate and confidence band is inspected via the Neyman-Pearson testing criterion~\citep{bouthillier2021accounting,neyman1928use}, where the X is statistically significantly better than Y if $P(X>Y) > 0.5 \; \land \; 0.5 \notin CI$.

\section{Extended Experiments and Discussion}
\label{app:eed}
\begin{figure}[h]
    \centering
    \includegraphics[width=\linewidth,trim=0cm 0cm 0cm 0cm,clip]{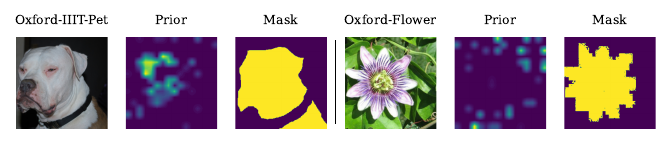}
    \caption{Comparison between pseudo-discriminative regions derived from an external prior (CLIP) and ground-truth discriminative regions obtained from segmentation masks for two datasets. Compared to the ground-truth, the pseudo-discriminative regions are (1) fragmented and incomplete (left panel), and (2) spurious (right panel). For brevity, we show two examples; however, similar patterns are observed throughout both datasets.}
    \label{fig:pseudo}
\end{figure}
\paragraph{Pilot study.}
To assess the trustworthiness of the reference regions projected by pretrained zero-shot models (\ie pseudo-discriminative regions), we conducted a pilot study with CLIP to visualize them.
In particular, we first selected two vision datasets with segmentation masks (\ie ground-truth discriminative regions), Oxford-IIIT-Pet and Oxford-Flower.
We selected these two datasets as they exhibit high inter-class similarity, where distinguishing between classes depends on subtle features.
Importantly, this setting allows us to assess whether CLIP---being trained on broad classification rather than spatial localization---is capable of capturing the fine-grained distinctions necessary to localize such similar classes.
Subsequently, similar to~\citet{yang2026learning},~\citet{nautiyal2025paric}, and~\citet{petryk2022guiding}, we derived pseudo-discriminative regions for images in the two datasets by passing them through CLIP and employing Grad-CAM on the image encoder.
In~\autoref{fig:pseudo}, we show that contrary to the assumptions of the aforementioned work, when compared to the ground-truth, these pseudo-discriminative regions can be (1) fragmented and incomplete (left panel), and (2) spurious (right panel).
Thus, the reference regions projected by CLIP are poor proxies for ground-truth discriminative regions and untrustworthy.

Furthermore, this finding motivates a broader question: where should vision models attend to when making decisions?
While some works suggest that vision models only require a subset of ground-truth discriminative regions for correct decisions~\citep{li2022expansion,fong2017interpretable}, other works contend that this can lead to over-reliance and reliable decisions require reasoning over the entire object region~\citep{zhang2018adversarial,Singh_2017_ICCV}.
In this work, we adopt the latter perspective, arguing that for attention supervision to be reliable, it needs to be grounded in trustworthy regions, that is, regions that are well-structured, complete, and accurate.

\paragraph{Limitations of recent work.}
Beyond their reliance on pseudo-discriminative regions, recent work in attention supervision exhibit several additional limitations~\citep{yang2026learning,nautiyal2025paric,petryk2022guiding}.
First, their regularization formulations are often sub-optimal, focusing solely on suppressing attention in spurious regions, and thus fail to capture the full extent of discrepancy between model attention and reference regions.
Second, they do not formally address scalability, particularly when employing gradient-based attention extraction techniques that incur significant overhead.
Third, they often do not quantitatively evaluate attention alignment and faithfulness even though their primary focus is to improve attention, but instead focus on evaluating generalization.
Fourth, their empirical evaluations lack statistical rigor with conclusions often drawn without formal statistical tests.
Finally, their applicability is largely restricted to standard supervised learning settings, with limited exploration in other settings.
Building on the foundations of recent work, our work addresses these limitations by introducing a principled regularization objective, a scalable formulation for it, and a statistically rigorous, attention-focused evaluation across multiple learning settings.

\begin{table}[h]
    \centering
    \begin{threeparttable}
    \caption{Comparison of Spearman and Pearson correlation scores (mean $\pm$ standard deviation) between regularizers. Results in bold denote cases where CAMAL achieves a higher mean than the suppress-only regularizer, with the difference being statistically significant under WSRT ($p<0.05$).}
        \begin{tabular}{@{}llrrrr@{}}
            \toprule
            \multirow{2.6}{*}{Dataset} & \multirow{2.6}{*}{Perturbation} & \multicolumn{2}{c}{Spearman} & \multicolumn{2}{c}{Pearson} \\
            \cmidrule(lr){3-4} \cmidrule(lr){5-6}
            & & \multicolumn{1}{c}{CAMAL} & \multicolumn{1}{c}{Supp. Only} & \multicolumn{1}{c}{CAMAL} & \multicolumn{1}{c}{Supp. Only}\\
            \midrule
            & Shift & \textbf{0.98 $\pm$ 0.07} & 0.50 $\pm$ 0.46 & \textbf{0.96 $\pm$ 0.05} & 0.59 $\pm$ 0.41 \\
            Oxford-IIIT-Pet & Erode & \textbf{1.00 $\pm$ 0.01} & --- & \textbf{0.99 $\pm$ 0.01} & --- \\
            & Dilate & 1.00 $\pm$ 0.01 & 1.00 $\pm$ 0.01 & 1.00 $\pm$ 0.01 & 1.00 $\pm$ 0.01 \\
            \cmidrule{1-6}
            & Shift & \textbf{0.97 $\pm$ 0.10} & 0.50 $\pm$ 0.45 & \textbf{0.94 $\pm$ 0.08} & 0.59 $\pm$ 0.39 \\
            Oxford-Flower & Erode & \textbf{1.00 $\pm$ 0.00} & --- & \textbf{1.00 $\pm$ 0.01} & --- \\
            & Dilate & 1.00 $\pm$ 0.02 & 1.00 $\pm$ 0.02 & 0.99 $\pm$ 0.03 & 0.99 $\pm$ 0.03 \\
            \bottomrule
        \end{tabular}
    \label{tab:cpr}
    \end{threeparttable}
\end{table}
\begin{figure}[h]
    \centering
    \includegraphics[width=0.9\linewidth,trim=0cm 0cm 0cm 0cm,clip]{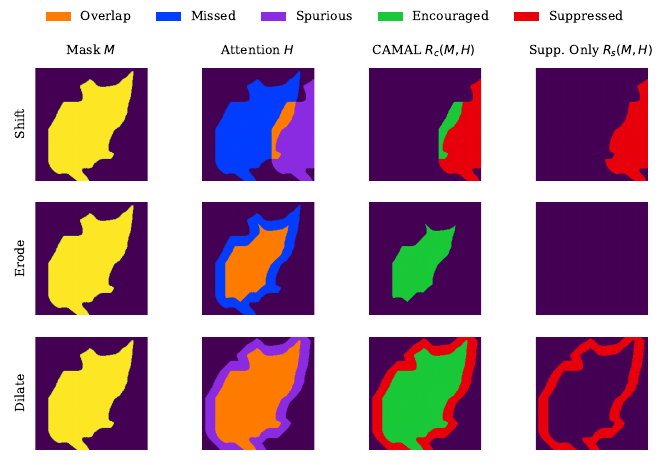}
    \caption{Comparison of regularizers under different perturbations---shift, erode, and dilate. For each perturbation we visualize the segmentation mask $M$, attention $H$, and corresponding regularizer responses $R(M,H)$.
    Colors denote overlap between ground truth regions and attention ($M \cap H$, orange), missed ground-truth regions ($M \setminus H$, blue), attention in spurious regions ($H \setminus M$, purple), with green and red indicating regions encouraged and suppressed by the regularizers respectively.
    A single representative example for each perturbation is shown for brevity.}
    \label{fig:cpr}
\end{figure}
\paragraph{Regularizer comparison.}
To compare the robustness and sensitivity of CAMAL's regularizer with the regularizers commonly used in recent work that focus solely on suppressing attention in spurious regions~\citep{nautiyal2025paric,petryk2022guiding}, we conducted a perturbation-based experiment that isolates their response to controlled deviations from ground-truth discriminative regions.
In particular, we perturb segmentation mask $M$ to simulate model attention $H$ and compare regularizers by evaluating $R(M,H)$.
We used three types of perturbations---spatial shifts, morphological erosion, and dilation---to systematically induce attention misalignment, under-coverage, and over-coverage respectively.
Similar to the pilot study, we conduct this experiment on the Oxford-IIIT-Pet and Oxford-Flower datasets, which provide high-quality segmentation masks for reliable evaluation.

For each image and perturbation type, we generate a sequence of perturbed masks $\{{H_i}\}_{i=1}^{n}$, progressively increasing in severity to simulate attention.
Subsequently, for each image, perturbation type, and regularizer $R$, we evaluate $R$ on all corresponding $(M, H_i)$ pairs to obtain sequence $\{R(M, H_i)\}_{i=1}^n$.
Lastly, we compute Spearman and Pearson correlations between perturbation severity and regularizer response.
We summarize the results in~\autoref{tab:cpr} and illustrate the behavior of each perturbation and regularizer in~\autoref{fig:cpr}.

Overall, the results in~\autoref{tab:cpr} show that \camal exhibits a consistent near-monotonic response to increasing perturbation severity, achieving strong correlations across all perturbation types.
In contrast, the \emph{suppress-only} regularizers used in recent work fail to respond meaningfully when attention deviates within ground-truth regions (\ie erode) and display weak or unstable correlations under spatial misalignment (\ie shift).
This behavior is further illustrated in~\autoref{fig:cpr}, where \camal simultaneously encourages attention within ground-truth regions and suppresses attention in spurious regions, while suppress-only regularizers do only the latter.
Thus, the suppress-only regularizers used in recent work lack any mechanism to recover missing evidence.
On the other hand, \camal captures the full spectrum of attention allocation, making it a fundamentally more complete and robust attention supervision objective.

\begin{proposition}[\textbf{Equivalence of batch-level attention extraction}]
    \label{prop1}
    Let $b \in \{1,...,B\}$ denote samples in a mini-batch of size $B$, let $y_b^{c_b}$ denote the scalar logit for the ground-truth class $c_b$ for sample $b$, and let $A_{b}^k$ denote the $k$-th feature map for sample $b$.
    Assume that the model processes each sample independently along the batch dimension so that for any $b \neq b'$, 
    \[
    \frac{\partial y_b^{c_b}}{\partial A^k_{b'}} = 0.
    \]
    Then, performing one backward pass on the summed scalar
    $\sum_{b} y_b^{c_b}$
    yields the same gradients as computing from each $y_b^{c_b}$ independently with $B$ backward passes.
\end{proposition}
\begin{proof}
    By the linearity of differentiation, the gradients of the summed scalar with respect to the $k$-th feature map of sample $b'$ can be expanded as
    \begin{equation}
    \frac{\partial}{\partial A_{b'}^k} \Bigg(\sum_b y_b^{c_b} \Bigg)    
    =\sum_b \frac{\partial y_b^{c_b}}{\partial A_{b'}^k}
    =\frac{\partial y_{b'}^{c_{b'}}}{\partial A_{b'}^k}
    ,\;\text{since}\;
    \frac{\partial y_b^{c_b}}{\partial A^k_{b'}} = 0 \; \forall\; b \neq b'
    \end{equation}
    Therefore, the gradient for each sample obtained from the summed scalar is identical to the gradient obtained from each individual scalar $y_b^{c_b}$ independently.
    Consequently, since attention is computed deterministically from these gradients and the corresponding feature maps, the resulting CAMs are identical under both approaches.
\end{proof}

\begin{table}[h]
    \centering
    \begin{threeparttable}
    \caption{Comparison of attention alignment between \camal and Vanilla in DL. The evaluation metric is test IoU (mean $\pm$ standard deviation). Results in bold denote cases where \camal achieves a higher mean than Vanilla, with the difference being statistically significant under WSRT ($p < 0.05$).}
        \begin{tabular}{@{}lrrrrrr@{}}
            \toprule
            \multirow{2.6}{*}{Model} & \multicolumn{2}{c}{BUSI} & \multicolumn{2}{c}{Oxford-IIIT-Pet} & \multicolumn{2}{c}{Oxford-Flower} \\
            \cmidrule(lr){2-3} \cmidrule(lr){4-5} \cmidrule(lr){6-7}
            & \multicolumn{1}{c}{CAMAL} & \multicolumn{1}{c}{Vanilla} & \multicolumn{1}{c}{CAMAL} & \multicolumn{1}{c}{Vanilla} & \multicolumn{1}{c}{CAMAL} & \multicolumn{1}{c}{Vanilla}\\
            \midrule
            ResNet & \textbf{33.9 $\pm$ 3.2} & 14.4 $\pm$ 2.4 & \textbf{67.2 $\pm$ 1.3} & 10.8 $\pm$ 0.9 & \textbf{70.0 $\pm$ 0.7} & 9.6 $\pm$ 0.9 \\

            ConvNeXt & \textbf{36.1 $\pm$ 2.2} & 2.0 $\pm$ 0.5 & \textbf{71.2 $\pm$ 1.1} & 3.7 $\pm$ 0.2 & \textbf{71.3 $\pm$ 0.6} & 3.4 $\pm$ 0.2 \\
            
            EfficientNet & \textbf{36.9 $\pm$ 2.2} & 12.8 $\pm$ 1.1 & \textbf{71.3 $\pm$ 0.6} & 8.8 $\pm$ 0.3 & \textbf{71.8 $\pm$ 0.5} & 9.7 $\pm$ 0.4 \\
            
            ViT & \textbf{44.0 $\pm$ 2.5} & 4.2 $\pm$ 3.2 & \textbf{79.9 $\pm$ 1.3} & 7.9 $\pm$ 8.3 & \textbf{80.8 $\pm$ 0.5} & 5.6 $\pm$ 3.1 \\
            
            Swin & \textbf{38.3 $\pm$ 1.6} & 9.2 $\pm$ 4.3 & \textbf{73.4 $\pm$ 0.4} & 8.9 $\pm$ 5.2 & \textbf{72.5 $\pm$ 0.4} & 3.6 $\pm$ 2.3 \\
            
            MaxViT & \textbf{38.9 $\pm$ 2.2} & 2.8 $\pm$ 2.5 & \textbf{73.9 $\pm$ 0.4} & 7.5 $\pm$ 1.9 & \textbf{71.7 $\pm$ 1.3} & 1.8 $\pm$ 0.5 \\
            \bottomrule
        \end{tabular}
    \label{tab:cva}
    \end{threeparttable}
\end{table}
\begin{figure}[h]
    \centering
    \includegraphics[width=0.9\linewidth,trim=0cm 0cm 0cm 0cm,clip]{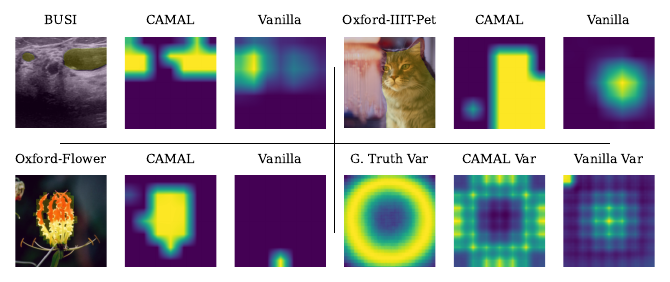}
    \caption{Comparison of attention alignment between \camal and Vanilla in DL. For each dataset, we present an image overlaid with its ground-truth discriminative regions, alongside the corresponding attention from \camal and Vanilla. A single representative example per dataset is shown for brevity; however, similar patterns are consistently observed across datasets.}
    \label{fig:cva}
\end{figure}
\begin{figure}[h]
    \centering
    \includegraphics[width=\linewidth,trim=0cm 0cm 0cm 0cm,clip]{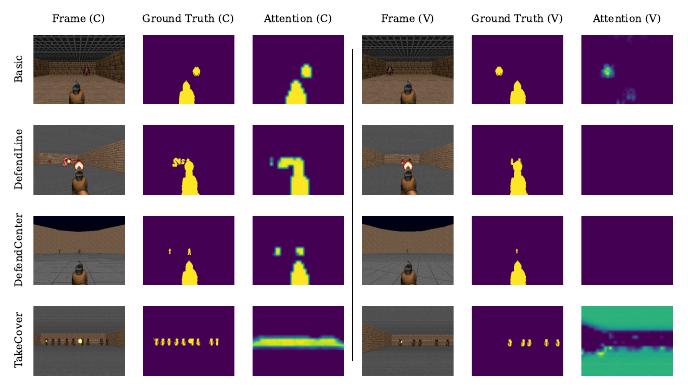}
    \caption{Comparison of attention alignment between \camal (C, left panel) and Vanilla (V, right panel) in DRL. For each environment, we present an image, its ground-truth discriminative regions, alongside the corresponding attention from \camal and Vanilla. A single frame for each environment is shown for brevity; however, similar patterns are consistently observed across episodes.}
    \label{fig:cvds}
\end{figure}
\paragraph{Attention alignment comparison with Vanilla.}
Quantitative and qualitative results for the DL experiments are shown in~\autoref{tab:cva} and~\autoref{fig:cva} respectively.
\camal consistently outperforms Vanilla across all models and datasets, achieving substantially higher IoU scores, with all improvements statistically significant under WSRT.
This difference is illustrated in~\autoref{fig:cva}, where \camal attends to the entire object region while Vanilla attends to incomplete or spurious regions.
This is further reflected in the bottom-right panel that visualizes the spatial variance of attention across datasets.
The variance by \camal closely resembles the ground truth, where variations occur primarily along object edges in a circular fashion.
In contrast, Vanilla exhibits variance in the top-left and center, suggesting that attention is often allocated to spurious regions or fails to adequately cover the object center.

Qualitative results for the DRL experiments are shown in~\autoref{fig:cvds}.
Across all environments, \camal produces attention that closely aligns with the ground truth, consistently highlighting relevant targets while suppressing background noise.
In contrast, Vanilla exhibits incomplete and misaligned attention, often deviating significantly from the ground truth.
For example, in TakeCover, \camal correctly ignores the background while Vanilla attends to the background.
These qualitative observations are consistent with the quantitative observations reported under SBCI and POI analysis in~\autoref{fig:cvda}.
Additional sample videos are provided in the source code to further illustrate the qualitative differences between \camal and Vanilla.
Overall, similar to the comparison with Prior, these results demonstrate that compared to Vanilla, \camal reliably yields better-aligned attention in both DL and DRL.

\begin{figure}[h]
    \centering
    \includegraphics[width=0.9\linewidth,trim=0cm 0cm 0cm 0cm,clip]{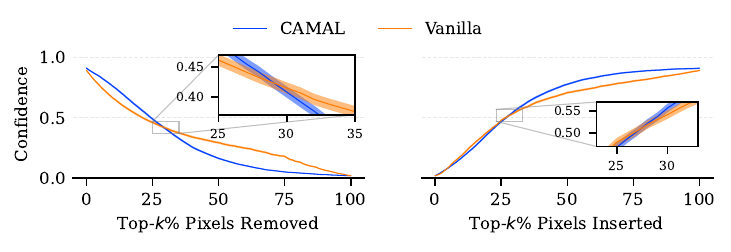}
    \caption{Comparison of attention faithfulness between \camal and Vanilla in DL. Model confidence is evaluated under removal and insertion of the top-$k$\% most important pixels determined by the model's attention. Curves are aggregated over all model-dataset combinations, with shaded regions denoting 95\% confidence bands.}
    \label{fig:cvf}
\end{figure}
\paragraph{Attention faithfulness comparison with Vanilla.}
The results for the DL experiments are shown in~\autoref{fig:cvf}.
The model confidence decreases and increases substantially faster for \camal for removal and insertion respectively.
This is further reflected in AUC, where \camal achieves 13.6\% lower AUC and 6.4\% higher AUC than Vanilla for removal and insertion respectively.
This indicates that the attention by \camal more accurately identifies features critical to its prediction.
Overall, similar to the comparison with Prior, these results demonstrate that compared to Vanilla, \camal reliably yields more faithful attention.

\begin{table}[h]
    \centering
    \begin{threeparttable}
    \caption{Mean training overhead ($\%$) incurred by \camal relative to Vanilla.}
        \begin{tabular}{@{}lrrr@{}}
        \toprule
             Model &  BUSI & Oxford-IIIT-Pet & Oxford-Flower\\
        \midrule
             ResNet &  1.0 & 5.0 & 0.2 \\
             ConvNeXt & 51.2 & 65.0 & 65.0 \\
             EfficientNet & 29.8 & 72.6 & 69.0 \\
             ViT & 55.2 & 63.1 & 62.7 \\
             Swin & 52.2 & 64.2 & 64.0 \\
             MaxViT & 60.5 & 69.2 & 68.9 \\
        \bottomrule
        \end{tabular}
    \label{tab:overhead}
    \end{threeparttable}
\end{table}
\paragraph{Overhead incurred by attention supervision.}
Attention supervision incurs overhead only during training, not at inference time, because the additional computation arises from computing attention and applying the auxiliary regularization---both of which are absent once the model is deployed.
Concretely, \camal requires one additional backward pass per mini-batch to compute attention, which constitutes the dominant source of overhead compared to Vanilla training.
As shown in~\autoref{tab:overhead}, this translates to negligible overhead for lightweight models (\eg ResNet) and moderate overhead for bigger models (\eg MaxViT) across datasets.

\paragraph{Limitations.}
A primary limitation of \camal is its reliance on segmentation masks, which are not universally available across all datasets.
While (1) many modern vision benchmarks increasingly provide such masks~\citep{illarionova2025packed,li2023behavior,OpenImagesV7,CelebAMask-HQ,Wydmuch2019ViZdoom,dosovitskiy2017carla}, and (2) recent advancements like prompt-based segmentation have significantly lowered the cost of generating human-verified masks~\citep{carion2025sam3segmentconcepts}, this requirement may still restrict applicability in domains where segmentation is unavailable or difficult to obtain.
In addition, while \camal does not introduce any additional inference costs, \camal introduces extra training costs due to the need to compute attention.
Taken together, \camal presents a trade-off: broader applicability may be constrained by mask availability and increased training cost, in exchange for more well-aligned and faithful attention, with improved or comparable generalization.

\paragraph{Broader impacts.}
A potential positive societal impact of this work is improved transparency and reliability of vision models, as better-aligned and more faithful models can help practitioners understand and trust model decisions in high-stakes domains such as healthcare and autonomous systems.
However, a potential negative impact is that such improvements could be misused to create more convincing AI systems whose decisions appear trustworthy, even when they are incorrect or biased, potentially enabling their deployment in sensitive or harmful applications. 
As a safeguard against potential misuse, we release only the source code and training procedures, but do not release pretrained model weights.

\section{Experimental Settings and Details}
\label{app:esd}
All experiments were conducted on a Ubuntu 22.04.5 LTS server equipped with an AMD EPYC 9554 CPU (64 cores, 128 threads), 755 GB of RAM, and four NVIDIA L40S GPUs, each with 48 GB of VRAM.
For both DL and DRL experiments, each trial was assigned a dedicated GPU to ensure fairness.
For example, the first trial for training the ResNet-BUSI-\camal DL combination had its own dedicated GPU.
The DL \camal-Vanilla comparison required 1.1K hours of computation, whereas the \camal-Prior comparison required 1.3K hours.
The DRL \camal-Vanilla comparison required 0.5K hours, bringing the total computational cost to approximately 3K hours.

\begin{table}[h]
    \centering
    \begin{threeparttable}
    \caption{The hyperparameters used for training the DL models.}
        \begin{tabular}{@{}lrl@{}}
        \toprule
             Hyperparameter & Value & Description\\
        \midrule
             Number of epochs & 500 & The number of epochs to train for. \\
             Learning rate & 0.001\tnote{1} & Optimized with AdamW~\citep{loshchilov2017decoupled}. \\
             Weight decay & 0.01 & The optimizer weight decay. \\
             Batch size & 32 & The mini-batch size to use. \\
             $\lambda_{\camal}$ & 1 & The \camal coefficient in~\autoref{eq:camalloss}. \\ 
        \bottomrule
        \end{tabular}
        \begin{tablenotes}
        \item[1] 0.0001 for ViT, Swin, and MaxViT
        \end{tablenotes}
    \label{tab:dlhyp}
    \end{threeparttable}
\end{table}
\paragraph{Deep learning.}
For the DL experiments, we selected six models spanning diverse architectural families, including convolutional-based, transformer-based, and hybrid convolution-transformer architectures (\eg ResNet, ViT, and MaxViT).
We use PyTorch's~\citep{ansel2024pytorch} default implementation of the evaluated models, initialized with ImageNet pre-trained weights~\citep{deng2009imagenet}.
Moreover, we selected three datasets with segmentation masks spanning two domains, medical imaging (BUSI) and natural images (Oxford-IIIT-Pet and Oxford-Flower).
BUSI is a dataset of breast ultrasound images containing annotated breast lesions, categorized into three classes: normal, benign, and malignant.
Oxford-IIIT-Pet and Oxford-Flower (also known as Flowers102) are datasets commonly used in vision tasks, comprising 37 pet categories (cats and dogs) and 102 flower categories respectively.
These three datasets share high inter-class similarity, making precise attention localization critical, thereby serving as a natural testbed for attention supervision.
The images in these datasets were resized to $256\times256$, center-cropped to $224\times224$, and normalized before being fed into the models for training.
We use 10-fold cross-validation when training each model-dataset-method combination.
The DL training hyperparameters can be viewed in~\autoref{tab:dlhyp}; for parameters not explicitly specified, we use the default settings by PyTorch.
For example, we use the default epsilon value of $1 \times 10^{-8}$ for the optimizer.

\begin{table}[h]
    \centering
    \begin{threeparttable}
    \caption{The hyperparameters used for training the DRL models.}
        \begin{tabular}{@{}lrl@{}}
        \toprule
             Hyperparameter & Value & Description\\
        \midrule
             Number of environments & 8 &  Number of environments used to collect training samples. \\
             Total timesteps & 10M & Total number of timesteps to train for. \\
             Number of timesteps & 128 & Number of timesteps per environment rollout. \\
             Number of epochs & 4 & Number of times to update the policy after each rollout. \\
             Number of mini-batches & 4 & Number of mini-batches in an epoch. \\
             Learning rate & 0.00025 & Optimized with Adam~\citep{kingma2014adam}. \\
             $\gamma$ & 0.99 & The discount factor. \\
             $\lambda_{GAE}$ & 0.95 & The general advantage estimation value. \\
             Clip coefficient & 0.1 & The surrogate clipping coefficient. \\
             Entropy coefficient & 0.01 & The entropy coefficient. \\ 
             Value function coefficient & 0.5 & The value function coefficient. \\ 
             Maximum gradient norm & 0.5 & The maximum norm for gradient clipping. \\ 
             $\lambda_{\camal}$ & 1 & The \camal coefficient in~\autoref{eq:camalloss}. \\ 
             Frame skip & 4 & The number of times an action is repeated. \\
             Frame stack & 4 & The history the model sees, inclusive of the current frame. \\
             Reward clip & -1, 1 & Lower and upper bounds for reward clipping. \\
        \bottomrule
        \end{tabular}
    \label{tab:drlhyp}
    \end{threeparttable}
\end{table}
\paragraph{Deep reinforcement learning.}
For the DRL experiments, we use CleanRL's PPO implementation~\citep{huang2022cleanrl} and selected four environments from the ViZDoom benchmark that cover diverse task demands: Basic, DefendLine, DefendCenter, and TakeCover.
Basic involves a single monster that the agent must eliminate, serving as a simple setting to evaluate fundamental perception and targeting.
DefendLine increases the difficulty by introducing multiple monsters aligned in front of the agent, requiring sequential target prioritization and efficient shooting.
DefendCenter further expands the complexity by spawning multiple monsters around the agent, demanding 360-degree awareness.
TakeCover differs from the previous tasks by focusing on survival rather than offense, where the agent must evade incoming enemy projectiles.
We train 5 agents for each environment-method combination.
The DRL training hyperparameters can be viewed in~\autoref{tab:drlhyp}; for parameters not explicitly specified, we use the default settings by CleanRL and ViZDoom.
For example, we use the default 5-layer CNN by CleanRL and the default frame size of $240\times320$ by ViZDoom.

\paragraph{Additional details.}
During training, we compute Grad-CAM at the final convolutional layer, as it captures high-level semantic information while preserving spatial structure; however this choice may vary depending on the model architecture (see source code for full implementation details).
Furthermore, we compute Grad-CAM with respect to the mean weighted combination of feature maps by dividing by $K$, the total number of feature maps, as we empirically observe that this yields more accurate CAMs (see ~\autoref{eq:gc}).

During evaluation, to compute the attention alignment for an image $I^c$ in the test set, we first quantize the corresponding model attention $H^c$ into a binary mask $H^c_{\mathrm{mask}}$ via simple thresholding, where values above $\tau$ are set to 1 and the rest to 0.
We use $\tau=0.7$, as we found it to effectively retain important regions while suppressing noise across models and datasets.
Subsequently, we compute the IoU between $H^c_{\mathrm{mask}}$ and the corresponding segmentation mask $M^c$ to obtain the attention alignment score for $I^c$:
\begin{equation}
    \text{IoU}(H^c_{\mathrm{mask}}, M^c)=\frac{|H^c_{\mathrm{mask}}\cap M^c|}{|H^c_{\mathrm{mask}}\cup M^c|}
\end{equation}
To evaluate attention faithfulness, we adopt a representative sampling strategy where, for each trial, we select one sample per class from the test set to perform removal and insertion-based faithfulness evaluation.
This approach ensures coverage across all classes while keeping the evaluation computationally tractable, as the removal and insertion experiments require multiple forward passes per sample.
Furthermore, the selected samples vary across different cross-validation splits, thereby mitigating sampling bias and improving the robustness of the aggregated results.

\paragraph{Assets and licenses.}
This paper is licensed under CC BY 4.0, while the accompanying source code is released under the MIT license.
The Oxford-IIIT-Pet dataset is distributed under the CC BY-SA 4.0 license, while the Oxford-Flower and BUSI datasets do not explicitly specify a license but are available for public use.\footnote{\href{https://www.robots.ox.ac.uk/~vgg/data/pets/}{https://www.robots.ox.ac.uk/~vgg/data/pets/}}\footnote{\href{https://www.robots.ox.ac.uk/~vgg/data/flowers/102/}{https://www.robots.ox.ac.uk/~vgg/data/flowers/102/}}\footnote{\href{https://scholar.cu.edu.eg/?q=afahmy/pages/dataset}{https://scholar.cu.edu.eg/?q=afahmy/pages/dataset}}
Moreover, our models use ImageNet pretrained weights, where ImageNet is publicly available for non-commercial research and educational purposes under its custom terms and conditions.\footnote{\href{https://www.image-net.org/download.php}{https://www.image-net.org/download.php}}
Finally, CleanRL and ViZDoom are released under the MIT license, while PyTorch is released under the BSD-3 license.\footnote{\href{https://github.com/vwxyzjn/cleanrl}{https://github.com/vwxyzjn/cleanrl}}\footnote{\href{https://github.com/Farama-Foundation/ViZDoom}{https://github.com/Farama-Foundation/ViZDoom}}\footnote{\href{https://github.com/pytorch/pytorch}{https://github.com/pytorch/pytorch}}
We use the latest available versions of all datasets; the exact software dependencies and their versions are provided in our source code.


\end{document}